\documentclass[conference]{IEEEtran}
\usepackage{tikz}
\usepackage{algorithm}
\usepackage{algorithmic}

\usepackage{graphicx}
\usepackage{subfigure}
\usepackage{latexsym,bm,amsmath,amssymb}
\usepackage{url}
\usepackage{amsthm}
\usepackage{xspace}
\usepackage{multirow}
\usepackage{color}
\usepackage{float}
\usepackage{tabularx}
\usepackage{caption}
\usepackage{tabu}

\def\ie{\textit{i.e.}\xspace}

\def\eg{\textit{e.g.}\xspace}

\def\BibTeX{{\rm B\kern-.05em{\sc i\kern-.025em b}\kern-.08em
    T\kern-.1667em\lower.7ex\hbox{E}\kern-.125emX}}

\newtheorem{theorem}{Theorem}

\theoremstyle{definition}
\newtheorem{definition}[theorem]{Definition}
\theoremstyle{remark}

\usepackage{endnotes}

\begin{document}
%
\title{VeriML: Enabling Integrity Assurances and Fair Payments for Machine Learning as a Service}

\author{\IEEEauthorblockN{Lingchen Zhao}
\IEEEauthorblockA{Wuhan University\\
lczhaocs@whu.edu.cn}
\IEEEauthorblockN{Chao Shen}
\IEEEauthorblockA{Xi'an Jiaotong University\\
chaoshen@xjtu.edu.cn}

\and
\IEEEauthorblockN{Qian Wang}
\IEEEauthorblockA{Wuhan University\\
qianwang@whu.edu.cn}

\IEEEauthorblockN{Xiaodong Lin}
\IEEEauthorblockA{University of Guelph\\
xlin08@uoguelph.ca}
\and

\IEEEauthorblockN{Cong Wang}
\IEEEauthorblockA{City University of Hong Kong\\
congwang@cityu.edu.hk}

\IEEEauthorblockN{Shengshan Hu}
\IEEEauthorblockA{Wuhan University\\
hushengshan@whu.edu.cn}

\and
\IEEEauthorblockN{Qi Li}
\IEEEauthorblockA{Tsinghua University\\
qli01@tsinghua.edu.cn}

\IEEEauthorblockN{Minxin Du}
\IEEEauthorblockA{The Chinese University of Hong Kong\\
dm018@ie.cuhk.edu.hk}
}

\maketitle

\begin{abstract}
Machine Learning as a Service (MLaaS) allows clients with limited resources to outsource their expensive ML tasks to powerful servers.
Despite the huge benefits, current MLaaS solutions still lack strong assurances on:
1) service correctness (\ie, whether the MLaaS works as expected);
2) trustworthy accounting (\ie, whether the bill for the MLaaS resource consumption is correctly accounted);
3) fair payment (\ie, whether a client gets the entire MLaaS result before making the payment).
Without these assurances, unfaithful service providers can return improperly-executed ML task results or partially-trained ML models while asking for over-claimed rewards.
Moreover, it is hard to argue for wide adoption of MLaaS to both the client and the service provider, especially in the open market without a trusted third party.

In this paper, we present VeriML, a novel and efficient framework to bring integrity assurances and fair payments to MLaaS.
With VeriML, clients can be assured that ML tasks are correctly executed on an untrusted server and the resource consumption claimed by the service provider equals to the actual workload.
We strategically use succinct non-interactive arguments of knowledge (SNARK) on randomly-selected iterations during the ML training phase for efficiency with tunable probabilistic assurance. We also develop multiple ML-specific optimizations to the arithmetic circuit required by SNARK. Our system implements six common algorithms: linear regression, logistic regression, neural network, support vector machine, K-means and decision tree. The experimental results have validated the practical performance of VeriML.
\end{abstract}

\section{Introduction}

Machine learning (ML) has been widely used in a variety of fields such as disease diagnosis, risk prediction, and pattern recognition. Since ML tasks often tackle with massive data, especially in the training procedure, they require servers with strong computational capabilities. As a result, Machine Learning as a Service (MLaaS) has become a promising service paradigm that enables weak clients to train ML models or compute predictions in powerful cloud infrastructures.

Despite the well-understood benefits, there exist many serious concerns regarding MLaaS practices. For wider adoption of MLaaS, especially in the open market without a trusted third party, we argue that the following three assurances are among the most important key desirables that MLaaS must satisfy.

- \textbf{Service correctness}:  The client needs to be ensured that the ML tasks done by the service provider must work as intended, as if they were done in-house, and always produce correct ML predictions or correct trained ML models.

- \textbf{Trustworthy accounting}: In existing commercial MLaaS platforms, the bill is normally based on the consumed computing resources~\cite{AWS,Azure}. Thus, the lack of full transparency at the service provider demands a strong assurance that the resource consumption claimed by the service provider indeed corresponds to the actual workload.

- \textbf{Fair payment}: The client should not obtain anything about the final MLaaS result (\eg, partial ML prediction result, sub-optimally trained ML models, etc.) prior to the payment. The fairness of exchange between the client's payment and the MLaaS result must be guaranteed.

Without these assurances, unfaithful service providers could return improperly-executed ML task results or partially-trained ML models while asking for over-claimed rewards. When designing schemes with strong assurance, we have to take into account all of the requirements together, so as to defeat sophisticated yet economically-incentivized attacks.
For example, the workload to train a model usually cannot be determined in advance since when the training converges is usually unpredictable.
The convergence depends on the training data set, the learning parameter setting, convergence conditions, and random factors in the nature of model training.
In this case, the server may claim that it has executed 10K iterations to train a model, but in fact it may have only executed 1K iterations. The client can only observe that the model is ``well-formed'', \ie, its dimension is consistent to the requirement, but doesn't know if such a model is actually produced by the 10K iterations.
Besides, a malicious server may also return inaccurate results by using a simpler model with fewer parameters, or even generating them randomly.
Therefore, verifying the correctness of the result as well as trustworthy accounting is vital and highly desirable.


In the literature, \emph{Verifiable Computation} (VC) is a powerful cryptographic building block which aims to verify the output of a deterministic function.
The worker can produce a proof of his computation by representing the function as a circuit.
An early attempt~\cite{ghodsi2017safetynets} invoked a common VC technique called interactive proof protocol to verify the prediction results of neural networks.
Yet, it does not address the cheating problem in MLaaS.
Firstly, although the VC technique can provide the zero-knowledge property, the server needs to reveal the results to the client for verification, which violates our premise of fair exchange.
Secondly, expressing the training process as a circuit straightforwardly is inefficient, due to the large number of training iterations and input variables.
Moreover, the matrix computations and non-linear functions in ML algorithms are also expensive for a circuit.

Hence, our goal is to design a fair machine learning service system which can ensure the correctness of the result with high efficiency. To transform this high-level concern into a practical system, we need to address three main challenges:
\begin{itemize}
  \item[1)] Because existing VC techniques need to express the computation as a circuit, how to efficiently convert the machine learning algorithms to circuits  is crucial for fast verification. We find that prior works in this regard are all unsuitable for VC, since their methods focus on the boolean circuit which is expensive in calculating  large numbers of multiplications~\cite{liu2017oblivious},~\cite{mohassel2017secureml}.

  \item[2)] How to avoid the expensive cost when loading all the iterations and input variables in the circuit? An intuitive idea is to sample and verify a small part of the iterations. However, naive sampling allows a malicious server to only execute the sampled iterations and ignore the remaining ones. Besides, due to the sequential nature of the training, dependencies cross training iterations must be taken into account to ensure the correct check of the sampled iterations.

  \item[3)] Ensuring that neither of the two parties can get advantages in the payment is essential to a practical MLaaS. Introducing a third party has been proved to be necessary, but this may cause additional trust concerns. Therefore prior works have proposed using blockchain to replace the third party~\cite{maxwell2011zero, campanelli2017zero, dziembowski2018fairswap}. However, directly invoking the VC technique in their protocols still presents risks of violating our fair exchange premise, where the training results from the server need to be revealed to the client for verification.
\end{itemize}

In this paper, we present VeriML, a practical outsourced machine learning framework, which can achieve integrity assurance and fair payment in training ML models, and  detect misbehaviors with high probability. Different from prior arts on verifiable computation, the client of VeriML can be assured of not only correct execution of ML tasks on the untrusted server, but also the fact that the claimed resource consumption by the service provider indeed corresponds to the actual workload. Unlike the prior method in~\cite{ghodsi2017safetynets}, we build VeriML by leveraging the generic VC technique--succinct non-interactive arguments of knowledge (SNARK), for its lower computation and communication costs for the client, better expressiveness, and zero-knowledge property~\cite{ben2013snarks,parno2013pinocchio}. Rather than directly applying SNARK to the entire ML task that often yields impractical costs, we focus on using SNARK only for randomly-selected individual iterative training procedures of ML for efficiency with tuneable probabilistic assurance.
To this end, we design a new commit-and-prove protocol, which avoids the expensive computation cost in proving  large numbers of loop iterations by sampling. Its correctness is ensured by our commitment design that preserves complex dependencies between different iterations without information leakage.

As case studies of VeriML, we focus on six typical machine learning algorithms: linear regression, logistic regression, neural network, SVM, K-Means and decision tree. To reduce the computation cost in generating the proofs, we propose multiple optimizations to construct more efficient arithmetic circuits required by SNARK. To the best of our knowledge, VeriML is the first solution that achieves fairness and correctness in outsourced machine learning services. Our contributions can be summarized as follows:
\begin{itemize}
  \item We present VeriML, the first outsourced machine learning system that can exchange results of a paid service fairly and detect computation misbehaviors  with high probability.
  \item To adapt machine learning algorithms to the VC technique, we design multiple circuit-friendly optimizations to verify expensive operations.
  \item We develop a new commit-and-prove protocol that can verify the loop iterations when training machine learning models with high efficiency. The detailed theoretic analysis demonstrates that our scheme can detect incorrect computations with high probability.
  \item We implement the VeriML framework with six popular ML algorithms: linear regression, logistic regression, neural network, support vector machine, K-means and decision tree on four real-world datasets to demonstrate its performance. Our results show that VeriML has a practical overhead as well as a negligible accuracy loss in training.
\end{itemize}

\section{Related Work}
\textbf{Verifiable Computation.} Generally, verifiable computation is commonly used to verify the correctness of a function without re-executing it. Previous studies have been focused on three mainstreams: authenticated data structures (ADS), interactive proofs (IPs) and succinct non-interactive arguments of knowledge (SNARK). Among them, ADS has limited expressiveness and lacks the zero-knowledge property, which can only process certain computations such as polynomial evaluation~\cite{papadopoulos2015practical} and graph queries~\cite{zhang2014alitheia}. IPs are implemented based on the sum-check protocol~\cite{lund1990algebraic}. They can solve practical problems such as multiplication matrix~\cite{thaler2013time} and SQL query~\cite{zhang2017vsql} with high efficiency by avoiding expensive cryptographic operations. But it is complex to convert an IP protocol to zero knowledge~\cite{cramer1998zero, chiesa2017zero, wahby2015efficient} as needed by us. Compared to the above two solutions, SNARK transforms an arbitrary polynomial-sized function to a circuit to produce a short proof~\cite{costello2015geppetto, wahby2015efficient, kosba2018xjsnark}, and its succinct property makes it very suitable for weak clients. SNARK supports zero-knowledge proof, and has rich expressiveness, while the high cost at the prover's side can be alleviated by the powerful cloud.

Recently, verification frameworks without using the generic VC techniques have also been proposed. For example, TrueBit~\cite{truebit} presents a novel peer-review idea which introduces smart contract to judge the correctness of the results. However, achieving the judge contract needs the blockchain to store the trained models provided by multiple verifiers, which may cause expensive costs.
Using trusted execution environment (TEE) such as Intel SGX is another orthogonal attempt, which might be complementary to our protocol~\cite{ohrimenko2016oblivious, tramer2018slalom}. Using TEE demands additional trust of hardware vendors in the first place, which does not always hold in MLaaS. Besides, the limited size of the current enclave may cause expensive overhead in secure I/Os and encryptions/decryptions, especially when addressing large-scale ML training tasks.

\textbf{Privacy-preserving Machine Learning.} Prior works on privacy-preserving ML over encrypted data usually adopted homomorphic encryption (HE) and garbled circuits (GCs). These techniques are similar to VC as they also need to represent the computation task. To our best knowledge, CryptoNets~\cite{gilad2016cryptonets} was the first to use neural networks to compute predictions on encrypted data, and subsequent works following this direction have further obtained results with improved performance and/or accuracy~\cite{liu2017oblivious,mohassel2017secureml,mohassel2018aby}.
The main difference between these works and ours lies in designing the circuit of machine learning algorithms, as we target at verifying the result. The most related work to ours is SafetyNets~\cite{ghodsi2017safetynets}, which verified the correctness of the prediction of neural networks using IPs. However, SafetyNets assumes that both the server and the client are aware of the model and results during the whole process (including prediction and verification), and does not concern the training phase. Therefore, we cannot leverage SafetyNets to implement the training process with integrity assurances and/or realize the fair exchange.

\textbf{Blockchain for Fairness.} Recently using blockchain to ensure fairness has  been broadly studied. Most of the related works aim to achieve fairness in multi-party computations by using timed commitment and garbled circuits such as~\cite{andrychowicz2014secure, bentov2014use, kumaresan2014use}, which are not suitable for our outsourced scenario. So far, zero-knowledge contingent payment (ZKCP)~\cite{maxwell2011zero} has been accepted as an elegant solution that uses zero-knowledge proof to ensure fairness in trading. Subsequent works~\cite{fuchsbauer2018subversion, campanelli2017zero, dziembowski2018fairswap} discussed the drawbacks of ZKCP when the buyer is malicious, or the purchase is digital service rather than goods, and also the expensive computation cost of verifying a large file. All the existing works still only concerned about how to prove ``if the seller holds something actually'', but without ensuring ``if what the seller holds is correct''.

\section{Preliminaries}
\subsection{Verifiable Computation (VC)}
The verifiable computation (VC) technique aims to enable a client $C$ to verify the correctness of function $F$ executed by server $S$, with a given input $x$.

For an outsourced task, the client first runs $\mathsf{KeyGen}$ to generate an evaluation key $EK_F$, and a verification key $VK_F$. It then sends  $EK_F$ to the server. The server executes $\mathsf{Prove}$ to produce the proof $\pi$, and sends $\pi$ and the result $y$ to the client. The client then checks the proof to verify the correctness of $y$ by executing $\mathsf{Verify}$ on the input $x$. The server and the client are referred to as the \emph{prover} and the \emph{verifier} respectively.

Technically, the implementation of proof generation relies on SNARK. Its key point is to encode the user-defined computations as quadratic programs. The basic flow is to first compile the program from a high-level language to an arithmetic circuit\footnote{We do not consider boolean circuit for verification because the numerous multiplications in machine learning algorithms make them more suitable for arithmetic circuits.}, then use the circuit to construct a Quadratic Arithmetic Program (QAP) which includes three sets of polynomials $A:=\{A_i(x)\}^m_{i=0}, B:=\{B_i(x)\}^m_{i=0}, C:=\{C_i(x)\}^m_{i=0}$ and a target polynomial $Z(x)$. Defining polynomial $P(x) = A(x)B(x) - C(x)$, and $Z(x)$ divides $P(x)$ iff $(c_1,\ldots, c_k)$ is a valid assignment for the circuit. The worker constructs $P(x)$ for the proof $\pi$, and the client can verify the correctness by checking if $Z(x)$ can divide $P(x)$. The \emph{zero-knowledge} property can be easily drawn into SNARK with a negligible overhead by choosing three additional random samples $\delta_1, \delta_2, \delta_3$ and adding $\delta_1Z(x), \delta_2Z(x), \delta_3Z(x)$ in the exponent to $A(x), B(x)$ and $C(x)$, respectively. The reader may refer to~\cite{gennaro2013quadratic} for more details.


\subsection{Fair Exchange Using Blockchain}\label{Fairness}
In a computation service, the \emph{fair exchange} problem is how to guarantee that the transaction between a seller $S$ and a buyer $B$ can be conducted fairly without one party cheating the other~\cite{maxwell2011zero, campanelli2017zero}. Fair exchange ensures: 1) the buyer who pays a potentially malicious seller can obtain the results; 2) the seller who delivers the results to a potentially malicious buyer can get paid.

Traditionally, a trusted third party is proved to be indispensable in attaining fair exchange~\cite{cleve1986limits}. It is, however, usually a serious limitation due to the lack of availability of a trusted third party. The emergence of blockchain makes it possible to achieve fair exchange without a trusted third party. Moreover, it is desirable that $B$ does not learn any knowledge about the goods except for what $S$ has published. To this end, Zero-Knowledge Contingent Payment (ZKCP) combines the hash-locked transaction and zero-knowledge proof technology. In the protocol, $S$ encrypts the file $f$ with the symmetric key $k$, and sends both $c=\mathsf{Enc}_k(f)$ and $s=\mathsf{SHA256}(k)$ with a proof $\pi$ to $B$. If $\pi$ can pass the verification algorithm run by $B$, $S$ can prove that $k$ is actually the key of $c$ and the preimage of $s$. Then $B$ posts a transaction on the blockchain to pay to anyone who can provide a preimage of $s$ to obtain $k$, and $S$ can obtain the money by presenting the preimage.

\section{Problem Statement}

\subsection{Definition}
In our system, a client $C$ outsources a machine learning task to a server $S$ with a training dataset $D$.
$S$ trains a prediction model $M$ according to a certain ML algorithm and parameters.
After the training phase, $C$ submits challenges $r$ to verify the execution of the learning algorithm, and in turn,
$S$ returns the corresponding proofs $\pi$ without providing $M$.
If all the proofs can pass the verification, $C$ is convinced that $S$ has faithfully completed the training, and then pays for the ML service to obtain the model $M$.
The core functionalities of our scheme are defined below.

\begin{definition}\label{Model Definition}
  \emph{A fair machine learning service system allowing a client $C$ to outsource the training algorithm $F$ and a dataset $D$ to the server $S$ is a tuple of five algorithms:}
\begin{itemize}
  \item $(EK_F, VK_F)\leftarrow\mathsf{KeyGen}(1^\lambda)$: \emph{is a probabilistic algorithm that takes as input a security parameter $\lambda$ and outputs a public evaluation key $EK_F$ and a verification key $VK_F$}.
  \item $(r, \mathbf{I})\leftarrow\mathsf{Compute}(D, F)$: \emph{is a deterministic algorithm that takes as input machine learning algorithm $F$ and a dataset $D$ and outputs the learning results $r$ and a commitment $\mathbf{I}$}.
  \item $(\pi, \mathbf{I}')\leftarrow\mathsf{Prove}(EK_F, F, x)$: \emph{is a deterministic algorithm that takes as input $EK_F$, an algorithm $F$, and data $x$, and outputs the corresponding proof $\pi$ and auxiliary information $\mathbf{I}'$}.
  \item $({0, 1})\leftarrow\mathsf{Verify}(VK_F, x, \mathbf{I}', \pi)$: \emph{is a deterministic algorithm that takes as input $VK_F$ and outputs $1$ if $F(x)=\mathbf{I}'$ and $\mathbf{I}' = \mathbf{I}$; $0$, otherwise.}
  \item $({0, 1})\leftarrow\mathsf{Payment}(c, u)$: \emph{is a deterministic algorithm that takes as input a rule $u$ and a condition $c$, and outputs $1$ if $c$ satisfies $u$; $0$, otherwise.}
\end{itemize}
\end{definition}

We say that VeriML is a secure protocol if the following properties are satisfied.
\begin{itemize}
  \item \textit{Completeness.} The probability that $\mathsf{Payment}$ outputs $1$ (Accept) is 1 if $\mathsf{Verify}$ outputs $1$ (Accept).
  \item \textit{Soundness.} The probability that $\mathsf{Verify}$ outputs $1$ (Accept) is less than $2^{-l}$ if $S$ does not follow the protocol, where $l$ is the bit length of the inputs.
  \item \textit{Fairness.} $C$ learns the witness iff he pays the fee, and $S$ gets paid iff he has the correct result.
\end{itemize}

The key problem to guarantee fairness between $S$ and $C$ in this outsourced computation service is to efficiently verify the result correctness.
The definitions of correctness, security and efficiency of verifiable computations~\cite{gennaro2013quadratic} are inherited.

\subsection{Threat Model}
We consider a generic setting for a cloud-based ML service, where a client $C$ uploads a training dataset $D$ to a server $S$, who runs the ML algorithm to train a model $M$. We assume that $S$ is malicious but rational, \ie, it may deviate from the protocol only if some additional economic benefits can be earned. Specifically, $S$ may cut back on the training process to save computation and storage costs. And we assume that if $S$ has indeed executed the task correctly, it will not deliver a fake model after the verification. In Section~\ref{Security}, we will discuss how to relax this assumption by verifying the performance of the delivered model. $C$ is considered to be honest-but-curious, \ie, it may try to learn the trained model $M$ before the payment. Obviously, \textit{verifiability} is critical to the system, \ie, the client is allowed to verify the correctness of the trained model without knowing any information about it.

The main purpose of our work is to solve the problem in verifying the integrity of ML model training, \ie, proving that $S$ has actually executed the specified computation task. Beyond that, we also aim to ensure the integrity of prediction services provided by $S$ without revealing the ML model to $C$. In the literature, there exist a series of works that study how to obtain information about the model or the server's training data by observing the prediction results, such as model inversion attack~\cite{fredrikson2015model}, model extraction attack~\cite{tramer2016stealing} and membership inference attack~\cite{shokri2017membership}. These studies are outside the scope of this work.

\section{Overview of VeriML}
In this section, we first discuss the main design challenges, followed by an overview of our VeriML system.

\subsection{Challenges}
A naive solution to the problem of verifiable outsourced ML is to invoke the existing VC protocols. The client $C$ can construct the circuit that covers the whole learning process, and generate the corresponding key pairs. The server $S$ evaluates the circuit and produces the proof. However, such a naive solution is infeasible in practice since ML algorithms usually require a large number of sequential loop iterations, and it is difficult to be represented by the circuit. Moreover, each iteration mainly consists of matrix multiplications and non-linear functions, which are extremely expensive for verification.

A straightforward modification to the above solution may be sampling several iterations for verification. A large enough sample size with uniformly-selected samples may help
detect incorrect iterations with high probability. However, due to the nature of training, it is difficult to ensure that the proved iterations are the same as required. For instance, the client may request the proof for the 10K-th iteration, but the malicious server might still be able to cheat by executing an iteration with an arbitrary input. This means that whether the produced proof is actually corresponding to the 10K-th iteration needs be further proved.

Another limitation of existing VC methods is that they require the server to reveal the trained model to the client for proof verification. This allows the client to obtain the trained model without payment. On the other hand, if the client does not own the trained model, the server can arbitrarily forge proofs to pass the verification. Now we are facing a dilemma: \textit{how can the correctness of the ML algorithm execution be verified without knowing the trained model?} And further, \textit{how can fair payments be guaranteed?}

\subsection{VeriML Outline}
The core idea of VeriML is to make the training process retrievable, namely, the verifier can reconstruct the inputs and outputs for any specific iteration, and the retrieval process is verifiable. After the verification, a fair exchange of the payment and the trained model can be done via the blockchain.
Figure~\ref{fig:overview} depicts the architecture of VeriML.
In summary, there are three phases: \emph{Computation Phase}, \emph{Verification Phase} and \emph{Payment Phase}.

\begin{figure}[t!]
  \centering
  \includegraphics[width = 0.96\columnwidth]{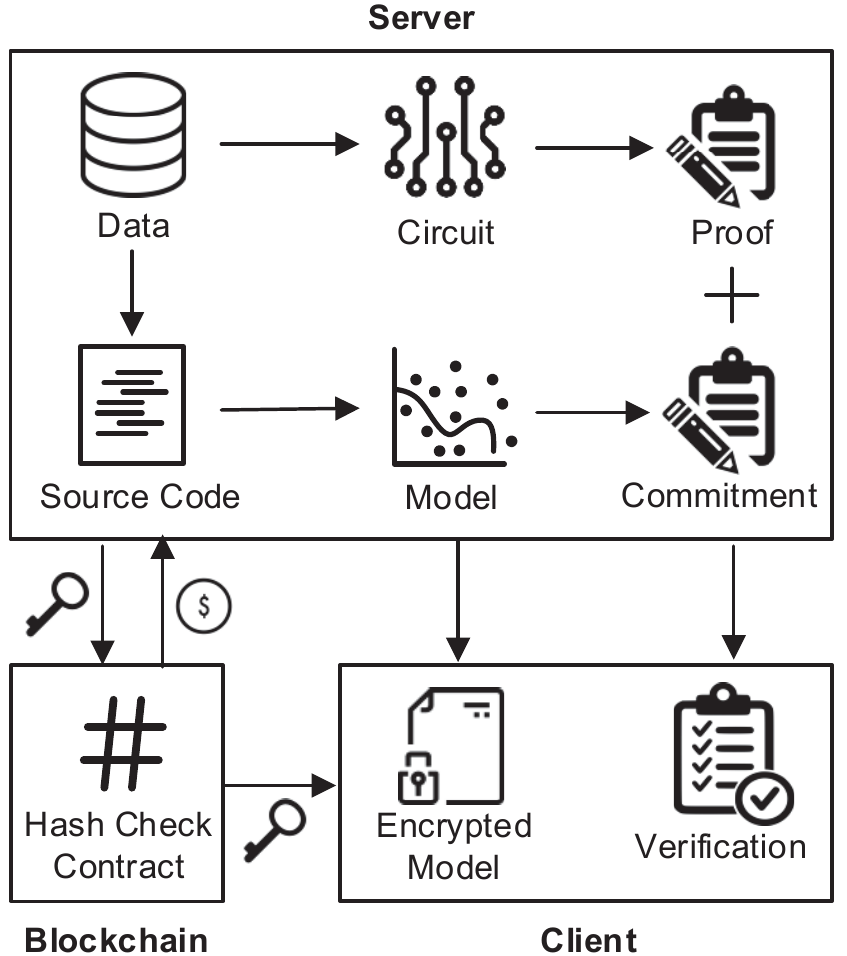}
    \caption{Overview and workflow of VeriML}\label{fig:overview}
\end{figure}

\textbf{Computation Phase.} $S$ executes the machine learning task. To enable the retrievability of the training process, $S$ needs to save additional auxiliary information, \eg, the intermediate states of the training process.

\textbf{Verification Phase.} After $S$ completes the training, $C$ sends multiple challenges to $S$ for the proofs of specified iterations. In particular, $C$ generates the key pairs and sends $EK_F$ to $S$. Instead of verifying all iterations, $C$ randomly samples a small subset of iterations as challenges, for which $S$ produces the corresponding proofs. If all proofs can pass the verification algorithm, $C$ considers that $S$ has executed the ML task correctly with high probability.

To retrieve a specified iteration, $S$ needs two inputs: the batch of the training data in the current iteration, and the model's state in the previous iteration. The data batch selection can be agreed upon by the two parties. The model's state can be retrieved by re-executing the training process. To reduce the additional overhead, $S$ saves some intermediate states of the model, referred to \emph{checkpoints}. Then any state can be retrieved from the checkpoints by several iterations without cryptographic operations. Finally, $S$ extracts a short identifier in each state during training, and sends it to $C$ for checking the correctness of the challenges.

\textbf{Payment Phase.} If all proofs pass the verification, $C$ will pay $S$ for the outsourced ML service. Inspired by ZKCP~\cite{maxwell2011zero}, we design a compact contract that only uses hash functions (\eg, $\mathsf{SHA256}$) for the limited blockchain-based scripting languages. The contract ensures a fair exchange between $C$ and $S$. Our protocol can resist some attacks of the original ZKCP~\cite{campanelli2017zero}, as will be discussed in Section~\ref{Security}.

\section{VeriML Design}
In this section, we present our design of VeriML. We first show how to provide a fair machine learning service for linear regression, and then extend the protocol to support other ML algorithms which include support vector machines, K-Means and decision trees.

\subsection{Linear Regression}\label{Linear}
We use a $d$-dimensional vector $\mathbf{w}$ to represent the parameters of the trained learning model. The value of $\mathbf{w}$ at the current iteration is referred to as a \emph{state}. In each iteration, the SGD algorithm inputs a batch of samples $\mathbf{X}$, their labels $\mathbf{Y}$ and the current state $\mathbf{w}$. The state $\mathbf{w}$ is updated by penalizing the deviation between predicted labels and actual labels of the batch. The update equation is
\begin{equation}\label{eq:BatchGD}
  \mathbf{w} = \mathbf{w}-\frac{\alpha}{b}\sum_{i=1}^{b}(\mathbf{x}_i\mathbf{w}-y_i)\mathbf{x}_i,
\end{equation}
where $b$ is the size of the batch and $\alpha$ is the learning rate.


\textbf{Decimal Arithmetic Operation.} Existing implementations of VC by QAP~\cite{parno2013pinocchio, costello2015geppetto, wahby2015efficient, libsnark} do not support decimal numbers in arithmetic operations. To address this problem, we adopt the fixed point number representation. We restrict that each input has at most $l$ bits of decimal points. Before invoking the circuit, we transform all inputs to integers by multiplying them by $2^l$. However, the multiplication of two inputs will yield a product that has a different length. In linear regression, the right side of Equation~(\ref{eq:BatchGD}) will become $\mathbf{w}\times2^l-\frac{\alpha}{b}\sum_{i=1}^{b}(\mathbf{X}\mathbf{w}\times2^{2l}-\mathbf{Y}\times2^l)\mathbf{X}\times2^l$ by the fixed point number representation, in which the amplification of each term is inconsistent, and the operation will fail to produce the right result. Therefore, we need to apply appropriate scaling factors in the equation by multiplying it with a certain constant. After each iteration, the output needs to be truncated to $l$ bits to serve as the input to the next iteration.

\textbf{Circuit Construction.} To verify the correctness of the SGD algorithm in each iteration, we need to construct a circuit for prediction, and then calculate the mean square error (MSE) of the prediction for the current batch of data. In linear regression, the prediction only involves addition and multiplication. MSE is calculated as $y=\frac{1}{b}\sum_{i=1}^{b}(\mathbf{w}\mathbf{x}_i-y_i)^2$, which can be implemented by addition and multiplication operations. If the batch size $b$ is more than 1, the division $\frac{1}{b}$ can be transformed to one multiplication by a constant since the batch size is set in advance, \eg, if the batch size is $32$, we can multiply the result by $2^{-5}$.

\begin{figure*}[t!]
  \fbox{%
  \begin{minipage} [t]{0.98 \linewidth}
   \noindent \textbf{VeriML Protocol for the Training Service}
     \begin{itemize}
        \item \textbf{Setup}:
        \vspace{-2pt}
        \begin{enumerate}
         \item The client $C$ sends the parameters of learning algorithm to $S$, which includes the learning rate $\alpha$, the batch size $b$ and the threshold of convergence $t$. After agreement, $C$ compiles the verification circuit $F$ to the QAP program locally.
         \item With the security number $\lambda$, $C$ generates the key pairs $(EK_F, VK_F)\leftarrow\mathsf{KeyGen}(F, 1^\lambda)$ and a random seed $s$ to be used for choosing batch samples. Then $C$ sends the seed $s$ to $S$.
        \end{enumerate}

        \item \textbf{Computation Phase}:
        \vspace{-2pt}
        \begin{enumerate}
        \item For each iteration $i$, $S$ selects the batch samples by seed $s$.
        \item Taking state $\mathbf{w}_{i-1}$ of the previous iteration and the data batch as input, $S$ outputs an updated state $\mathbf{w}_i$ and saves its identifier $\mathbf{I}_i$.
        \item After every $N/k$ iterations, $S$ saves the state as a checkpoint.
        \item If the difference of accuracy between epoch $e_m$ and $e_{m-1}$ is less than the threshold $t$, $S$ terminates the training process and sends identifiers $\mathbf{I}_1, \ldots, \mathbf{I}_{m}$ to $C$.
        \end{enumerate}

        \item  \textbf{Verification Phase}:
        \vspace{-2pt}
        \begin{enumerate}
        \item $C$ sends the circuit and $EK_F$ to $S$.
        \item $S$ checks the correctness of the circuit.
        \item $C$ randomly chooses $m$ iterations $s_1, \ldots, s_m$ for verification, and sends the set of indexes to $S$ to challenge the corresponding proofs $\pi_{s_1}, \ldots, \pi_{s_m}$.
        \item $S$ locates the checkpoints and retrieves states $\mathbf{w}_{s_1-1}, \ldots, \mathbf{w}_{s_m-1}$. Then $S$ produces proof $(\mathbf{I}'_{s_{i-1}}, \mathbf{I}'_{s_i}, \pi_{s_i})\leftarrow\mathsf{Prove}(EK_{F}, {F}, \mathbf{w}_{s_{i-1}}, \mathbf{X}_{s_i})$ for each $i$, and sends the results $(\mathbf{I}'_{s_{i-1}}, \mathbf{I}'_{s_i}, \pi_{s_i})$ to $C$.
        \item $C$ selects the inputs of samples also by the random seed $s$, then runs $v_{s_i}\leftarrow\mathsf{Verify}(VK_{F}, \mathbf{X}_{s_i}, \mathbf{I}'_{s_{i-1}}, \mathbf{I}'_{s_i}, \pi_{s_i})$ for all selected iterations. If there exists $v_{s_i}=0$, $C$ outputs $\mathsf{Reject}$; otherwise, $C$ compares $\{\mathbf{I}'_{s_{i-1}}$, $\mathbf{I}'_{s_i}\}$ with $\{\mathbf{I}_{s_{i-1}}$, $\mathbf{I}_{s_i}\}$ in the commitment. If $\forall$ $1\leq i\leq m$, the commitments and the verification results are the same, $C$ outputs $\mathsf{Accept}$, otherwise $\mathsf{Reject}$.
        \end{enumerate}

        \item \textbf{Payment Phase}: If $C$ outputs $\mathsf{Accept}$ after the verification phase, then the two parties enter the payment phase.
        \begin{enumerate}
        \item $S$ encrypts the trained model $\mathbf{w}$ by a symmetric key $k$, then sends $\textsf{Enc}_{k}(\mathbf{w})$ and $h=\textsf{SHA256}(k)$ to $C$.
        \item $C$ posts a transaction $T$ on the blockchain to pay the pre-determined fee $f$ to the party who presents $x$ such that $\textsf{SHA256}(x)=h$.
        \item $S$ presents string $z$ to $T$. If $\textsf{SHA256}(z)=h$, then $T$ sends the fee to $S$, otherwise returns the payment to $C$.
        \end{enumerate}

     \end{itemize}
  \end{minipage}
    }
  \caption{The VeriML protocol for the training service}
  \vspace{-10pt}
\label{VeriML}
\end{figure*}

\textbf{Model Blinding.} In existing VC schemes, $C$ must know the input and output of each iteration, \ie, the intermediate state of the model, to perform verification. However, the state cannot be revealed to $C$ since he may leverage it as a trained model and abort without payment. Therefore, we have to enable $C$ to verify the correctness of each iteration without knowing the input and output states of the iteration. The input state can be preserved by the zero-knowledge property of SNARK, and we design a blinding scheme to mask the output state as an \emph{identifier}. The blinding scheme should satisfy two requirements: 1) It is hard to retrieve the output state from the identifier and 2) The collision probability of two different output states is small. The second requirement prevents $S$ from producing a correct identifier with an incorrect input. Specifically, because the output state $\mathbf{w}'$ is only determined by the input state $\mathbf{w}$ and the batch $\mathbf{X}$, we need to ensure that it is hard for $S$ to generate a fake state which can pass the verification with a given batch.
	
Intuitively, the hash function is a good choice for blinding the states. Nonetheless, implementing hash functions such as $\textsf{SHA256}$ in arithmetic circuit requires a lot of shifting and bitwise operations, which are much more expensive than arithmetic operations, and it significantly raises the computation cost of $\mathsf{Prove}$. Moreover, the computation cost of hashes increases linearly with the dimension of the input state. To deal with this issue, we design a simple yet effective solution. We sum all elements in the output/input state vector $\mathbf{w}$, then only hash the summation\footnote{To avoid the additional leakage by the hash value, we concatenate the summation with another $254$-bit random number $r$ which is generated by a common random seed before hashing.}. The field of SNARK is $254$-bit long, shorter than the size of half a block in $\textsf{SHA256}$ (512 bits), thus the circuit only needs to execute one block in $\textsf{SHA256}$ to avoid the large cost in hashing a large state.

However, the straightforward summation has one shortcoming: the correctness of the summation result does not always ensure the correctness of the individual model parameters. To address this problem, we introduce an enhanced approach by using a random coefficient vector $v$ generated by $C$ to calculate the (random) weighted summation. And we would hash the inner product $\mathbf{w}\cdot\mathbf{v}$ as the identifier. We will analyze the security of this design in Section.~\ref{Security}.

\textbf{Hyper-parameter Encoding.} The learning rate $\alpha$ has a significant impact on the training performance.
Tuning $\alpha$ will affect the convergence speed of training, thereby affect the number of iterations and the model's performance. In this paper, we mainly consider a fixed learning rate such that we can encode $\alpha$ into the circuit as a constant to reduce the computation cost. If $\alpha$ changes from iteration to iteration, to avoid recompiling the circuit and generating the key, we can treat $\alpha$ as an input to the circuit, which adds $d$ multiplications for one iteration. The batch size $b$ is another important hyper-parameter that can be encoded into the circuit. Large batch size will increase the computation cost in each iteration. If the batch size increases by $1$, $2d$ additional multiplications are required for each iteration, which greatly augments the proving time.

\textbf{Iteration Verification.} After constructing a circuit to verify the execution of each iteration, $C$ can check all iterations to validate the integrity of the whole training process. However, the number of iterations may be huge, which leads to a prohibitively long time for proving correctness. To tackle this issue, we design a novel protocol which randomly samples a small subset of iterations to detect the misbehaviors during training.

Let $N$ denote the total number of iterations that are claimed to be completed by $S$. In the training process, after finishing the $i$-th iteration ($1\leq i\leq N$), $S$ saves the identifier $\mathbf{I}_i$ (a hash value of the summation of the elements in the output state $\mathbf{w}_i$). Upon completing the training process, $S$ sends all identifiers $\mathbf{I}_1, \ldots, \mathbf{I}_N$ to $C$ as a commitment for verification. $S$ may adulterate identifiers to falsify his workload.

When $C$ has received the commitment, he randomly samples a small subset $\{s_1, \ldots, s_m\}$ of $m$ iterations ($1\leq m\leq N$) as challenges. To produce the proof of the $s_i$-th iteration, $S$ needs two inputs: the output state $\mathbf{w}_{s_{i-1}}$ of the previous iteration and the batch $\mathbf{X}_{s_i}$ of this iteration. Then $S$ runs $\mathsf{Prove}(EK_F, {X}_{s_\mathbf{I}}, \mathbf{w}_{s_{i-1}})$ to produce the proof and sends it to $C$. Finally, $C$ checks if the produced proof $\pi_i$ can pass $\mathsf{Verify}(VK_F, \mathbf{I}_{s_{i-1}}, \mathbf{I}_{s_i}, \pi_{s_i})$  and determines if $\mathbf{I}_{s_i}$ is the same as the corresponding identifier in the commitment.

Apart from verifying the consistency of the output of $\mathsf{Verify}$ with the commitment, we also have to verify the authenticity of the input state because $S$ may output meaningless results but still can pass the verification. For example, $S$ can use an all-zero vector $\mathbf{0}$ as a fake state. The output of forward propagation is always $\mathbf{0}$, which allows $S$ to generate correct identifiers while cheating on the computation (only conduct the backward propagation calculation). However, the identifiers in the commitment are generated by the output of each iteration, which is different from the input to the next iteration, since we truncate the output to $l$ bits before feeding it as input to the next iteration. Therefore, we cannot verify the authenticity of the input state directly from the commitment. To solve this problem, we propose to let $S$ prove that it has the preimage of identifier $\mathbf{I}_{{i-1}}$ in the commitment, and the difference between the preimage and the summation of the input state is less than $\frac{d}{2^{l}}\sum_{j=1}{d}v_{ij}$, where $d$ is the dimension of the data.

Saving all intermediate states inevitably induces a high storage cost for $S$. Therefore, we propose to set checkpoints to retrieve the states. Specifically, we partition the $N$ iterations into $k$ groups, and thus each group contains $N/k$ states. $S$ only saves the output state of the first iteration in each group as a checkpoint such that the $i$-th state in a group can be retrieved by re-conducting $i$ iterations from the checkpoint. Note that $S$ can tune the parameter $k$ to make a tradeoff between the storage and retrieving costs. This retrieving process does not involve any cryptographic operations, \ie, generating keys or producing proofs, thus the additional computation cost is acceptable.

Obviously, the core of our scheme is to reproduce the identifiers by VC. To achieve this, in the training and verification phases, the chosen data batches need to be consistent. In general, the inputs are chosen by randomly shuffling the whole dataset, so we can let $C$ assign a random seed and send it to $S$. For each epoch, as long as the two parties use the same seed and algorithm to shuffle the dataset, they can reproduce and verify the identifiers later. Note that $C$ has to check the distinctness of the selected batches. Otherwise, $S$ will have advantages to forge identifiers.

  \begin{figure}
    \fbox{%
    \begin{minipage}{0.98 \linewidth}
     \noindent \textbf{VeriML Protocol for Prediction Service}
          \begin{enumerate}
           \item $C$ generates the key pair $(EK_F, VK_F)\leftarrow\mathsf{KeyGen}(F, 1^\lambda)$, then $C$ sends $EK_F$ and data $\mathbf{X}$ to $S$.
           \item $S$ runs the prediction algorithm to compute the result. Then $S$ produces the proof $\pi\leftarrow\mathsf{Prove}(EK_{F}, {F}, w, \mathbf{X})$ which takes the result as a witness, and compares if the witness is equivalent to the output of the circuit, then sends the produced proof $\pi$ to $C$.
           \item $C$ runs $v\leftarrow\mathsf{Verify}(VK_{F}, \mathbf{X}, \pi)$ and outputs $\mathsf{Reject}$ if the output is 0. Otherwise, $C$ sends the request for the result.
           \item $S$ encrypts the result $m$ as $c$ by a symmetric key $k$: $c=\textsf{Enc}_{k}(m)$, and sends $c$ and the hash value $h=\textsf{SHA256}(k)$ to $C$.
           \item $C$ posts a transaction $T$ on the blockchain to pay the fee $f$ to the party who presents $x$ such that $\textsf{SHA256}(x)=h$.
           \item $S$ presents the string $z$ to $T$. If $\textsf{SHA256}(z)=h$, $T$ sends the fee to $S$, otherwise gives the refund to $C$.
          \end{enumerate}
    \end{minipage}
    }
    \caption{VeriML protocol for the prediction service}
    \vspace{-10pt}
  \label{VeriMLp}
  \end{figure}

\textbf{Payment.} The payment process has to ensure a fair exchange, \ie, to prevent either $S$ or $C$ from cheating one another. More specifically, a malicious $C$ would like to obtain the trained model without paying anything while a malicious $S$ expects to get the payment without executing the whole training process. Inspired by ZKCP, we design an effective payment protocol that can achieve a fair exchange. First, $S$ encrypts the final state of the trained model using a symmetric key $k$, then sends $\textsf{Enc}_k(\mathbf{w})$ and $h=\mathsf{SHA256}(k)$ to $C$. Then, $C$ posts a transaction on the blockchain to transfer the payment to anyone who reveals a preimage of $h$. In this way, $C$ cannot decrypt the trained model without payment via the transaction on the blockchain, and $S$ cannot obtain the payment without passing the verification or providing the key to $C$. We hereto realize the fair exchange.

\textbf{Prediction Service.} We also consider the case that $S$ provides the prediction service to $C$ instead of training a model. To verify the prediction results, we can construct a similar circuit, but do not have to yield multiple identifiers to help check the correctness of each iteration. Here, $S$ should not directly reveal the prediction result to $C$ for verification since $C$ may get away with the result without paying anything. Hence, we transform the verification process to a zero-knowledge proof that sets the prediction results as a witness, and compares whether the witness equals the result of forward propagation calculated by the circuit. If it does, $C$ can confirm that $S$ actually generates the correct prediction results.

We present the protocol for the training service by summarizing the above constructions in Figure~\ref{VeriML}. An adapted protocol for the prediction service is given in Figure~\ref{VeriMLp}.

\subsection{Logistic Regression}\label{Logistic}
Compared with linear regression, logistic regression faces the main challenge of computing the sigmoid activation function $f(x) = \frac{1}{1+e^{-x}}$, since the division and exponentiation are not supported by an arithmetic circuit. Prior works have presented various approaches to approximate the sigmoid function. In this section, we will discuss the differences of these approaches in terms of efficiency and accuracy, and propose a more feasible solution for the QAP application.

Piecewise and polynomial approximation are two mainstream approaches to implement sigmoid functions in prior works~\cite{mohassel2017secureml, gulcehre2016noisy, cheon2017privacy}. The piecewise method, which relies on comparisons, is much more expensive than arithmetic operations in QAP, since one comparison between two $l$-bit integers requires a \emph{split} operation that consumes $l+2$ constraints while one multiplication between two integers only consumes one constraint. Therefore, the polynomial approximation is obviously more favorable.

Taylor expansion is a classical method to approximate nonlinear functions, and its accuracy highly depends on the degree of the polynomial terms. The higher the degree of the polynomial term, the better the approximation performance. But the computation cost will be higher, and it is easy to exceed the finite field. Inspired by~\cite{cheon2017privacy}, we use the \emph{Remez} algorithm to implement the approximation with high efficiency and accuracy. However, the Remez method is still only suitable for a certain range of the input because of the unbounded tails, and the input beyond the proper range may affect the accuracy of the approximation, hence we need to set an appropriate range to calculate the approximated polynomial.

Setting the degree of polynomials as three, and the approximated range of $x$ as [-5, 5], the approximation of Remez is $f(x)=-0.004x^3+0.197x+0.5$ and that of the Taylor expansion is $f(x)=-\frac{1}{48}x^3+0.25x+0.5$. It can be seen that the Remez-based approximation is closer to the original sigmoid function in a wider range. More results about the training accuracy of the approximations are presented in Section~\ref{sec:eva}.

\subsection{Neural Networks}\label{NN}
To efficiently apply VeriML to neural networks, we design an \emph{inversed verification} method to reduce the size of the circuit by utilizing pre-computed results.

Traditional matrix multiplications using circuit has a complexity of $O(n^3)$,
which is very time-consuming. Observations in prior works show that verifying the correctness of results is much cheaper than computing the results forwardly~\cite{ben2013snarks, costello2015geppetto, kosba2018xjsnark}. For example, verifying $c=a/b$ is hard to be implemented in QAP but instead we can verify $a=b\times c$ efficiently. Following this rationale, we use Freivald's algorithm~\cite{motwani2010randomized} to inverse the forward computation when constructing the circuit. Freivald's algorithm is a probabilistic randomized algorithm for verifying matrix multiplications. Assume that we have three $n\times n$ matrices $A$, $B$ and $C$. By using a uniformly-sampled $n\times 1$ random vector $r$ over field $\mathbb{Z}_s^n$, the correctness of $AB=C$ can be reviewed by verifying whether $A(Br)=Cr$ stands. The false positive rate is $1/(s+1)$. By selecting a large field, we can reduce the probability of a false positive to a negligible value. Because the random $r$ is selected after $S$ has generated the commitment, it can be selected in advance and encoded in the circuit as constants to save the computation cost.

With the help of Freivald's algorithm, the verification requires only $O(n^2)$ multiplications. For neural networks, let $n_i$ denote the number of neurons in the $i$-th layer. The number of matrix multiplications in one layer descends from $bn_in_{i-1}$ to $bn_{i-1}$ (with $bn_{i-1}+n_in_{i-1}$ multiplications by constants). For a neural network with $\beta$ matrix multiplications, as these multiplications are independent, the false positive rate for one iteration is less than $1/(s+1)^\beta$, which is negligible.

In a typical SNARK implementation, the arithmetic circuit operates over a 254-bit field. The continuous multiplications of multiple layers make the length of the results rapidly increase, which may lead to an overflow. The proposed verification inversion method can mitigate the overflow due to fewer multiplications. Furthermore, we can truncate the multiplication results before feeding them to the circuit. To avoid the problem of inconsistent bit lengths of different results, we conduct an additional check to see if the difference between the two results is less than $2^{-l}$. In addition, some other tricks which utilize randomness in training neural networks, such as dropout, are easy to be implemented by using a consistent random seed.

\textbf{Softmax Function Verification.} The softmax function $f(x_i)=\frac{e^{-x_i}}{\sum_{i=1}^{\kappa}e^{-x_i}}$ is used for multi-class classification, where $\kappa$ is the number of classes. To verify the function, we have to tackle the difficulties of computing exponentiation and division in the circuit. The output of softmax is a probability distribution so that all results are non-negative. Therefore, we can adopt square function to replace exponentiation operations. Unlike other activation functions that can be approximated into division-free forms, the division operation in the softmax function is inevitable. Since that division is not supported by SNARK (division by constant can be transformed to multiplication), we adopt the strategy of inverting verification that checking the equality of the dividend and the product of the divisor and the result. For each data sample, $S$ calculates the results of softmax and feeds them as witness to the circuit. For a batch of data, the input to the circuit is a $b\times\kappa$ matrix, and the circuit needs to perform $bm$ additional multiplications.

\subsection{Support Vector Machine}
Besides the SGD-based methods, support vector machine (SVM) is another classic and popular machine learning algorithm, for solving classification problems. The training process of SVM can also be represented as sequential loop iterations, so we can transfer it into the VeriML framework naturally. In this paper, we consider the most representative case of binary classification proposed in~\cite{shalev2011pegasos}. The details of SVM are omitted due to the space limitation.

The main construction of training a SVM model is the same as the prior methods, \ie, using a small batch of data to update the model after each iteration until the objective function converges. Specifically, apart from the basic additions and multiplications, each iteration includes two divisions, one Euclidean projection and $b+1$ comparisons. As discussed before, verifying divisions and square roots can be transformed to the multiplication operations using the pre-computed results. Therefore, we can construct the arithmetic circuit straightforwardly.

\subsection{K-Means}
VeriML can be scaled with clustering, a type of task which expects to partition multiple data samples into several clusters, because the process of training a clustering model also consists of sequential loop iterations. Here we use the most typical clustering algorithm--K-Means to demonstrate the protocol design.

In the beginning, $C$ randomly chooses $k$ centroids to represent the clusters. In each iteration, the K-Means algorithm assigns each training sample to the cluster closest to it, then uses the average of all the samples in this cluster to update the centroid. Here we choose to use a small batch of data in each iteration to reduce the circuit size. According to the prior results~\cite{sculley2010web}, this operation will not have a significant impact on the accuracy. Because the centroids are represented by vectors, the commitment can be generated by the random coefficient vector as well.

\textbf{Verifying the Closest Distance. }The main cost of executing the circuit lies in finding the closest distance of the $k$ centroids. When the batch size is $b$, each iteration needs to execute $bk$ comparisons, which involve large computation overhead. Here we avoid the comparisons by checking the correctness of the candidate closest distance given by the server.  Specifically, for each data sample, the circuit takes the candidate closest distance as an input. Then the circuit executes the subtractions between the candidate closest distance and all the $k$ previously-obtained distances, respectively. If the circuit finally outputs only one 0, and all the other results are negative, the candidate closest distance is considered as a correct one.

\subsection{Decision Tree}
Decision tree and its variants are significantly different from other machine learning methods discussed above. This is because training or using a tree is composed of comparisons, rather than additions and multiplications. Intuitively, verifying the correctness of the structure of a trained tree, is equivalent to checking two conditions: 1) for each internal node, whether the partition is done based on the largest information gain, and 2) whether the data samples belonging to this node are actually composed by its children nodes. The sampling strategy can still be adopted to verify the integrity of the whole tree, and the comparisons can be verified efficiently by the technique we proposed in verifying K-Means. Because the number of data samples directly affects the correctness of finding partitions, the batch strategy is no longer applicable. Now we are facing two challenges: (1) how to reduce the circuit's I/O caused by traversing all the training data, and (2) how to avoid revealing the tree to $C$ in the verification phase.

\textbf{Compressing the Inputs.} Instead of using batches to split the data horizontally, we use histograms to reduce the input size vertically. LightGBM~\cite{ke2017lightgbm} proposed representing inputs with histograms to accelerate the training of decision trees. Inspired by this idea, in our design, $C$ buckets feature values into multiple bins before the training process. Concretely, for each feature, $C$ first converts its field to bins, then traverses all the data to construct the histogram and sends it to $S$. $S$ can only use the histogram to calculate information gains and find partitions. Assume that there are $n$ data samples. Each sample has $d$ features, each of which is bucketed into $k$ discrete bins. Using histogram thus can reduce the number of inputs from $nd$ to $kd$. In the meanwhile, verifying if a node is correctly partitioned can be transformed to verifying if its histogram is equivalent to the addition of its child nodes' histograms.

\textbf{Commitments for Decision Tree.} Compared with other ML algorithms, training decision tree does not include the iterative optimization process. This makes the previous method which blinds the states of the model no longer applicable. Instead, we observe that the histogram is updated after each partition, and the updates cannot change  the structure other than the values of bins, so we can blind the histogram with the same commitment method described before, and check if the histograms of sampled nodes are the same with the commitment. Such a commitment may reveal the structure of the tree to $C$, and this can be fixed by adding dummy nodes to fill the tree to be ``perfect'', \ie, all the leaves have the same depth and all the internal nodes have the same degree.

Computing predictions by decision tree can only be implemented by multiple comparisons. Note that for categorical features, comparing the features with partitions can be transformed to checking the equalities, which has lower computation costs.

\section{Security Analysis}\label{Security}

\begin{theorem}
Assuming there exists no more efficient algorithm that can output the correct result than the training algorithm, the proposed protocol VeriML is secure if the properties of completeness, soundness, and fairness are satisfied simultaneously.
\end{theorem}

\begin{proof}
This assumption is reasonable, since otherwise the original training algorithm would be meaningless.

\emph{Completeness.} The completeness of the setup, computation and verification phases depends on the completeness of the underlying verifiable computation scheme.
The completeness of the payment phase depends on the correctness of $\textsf{SHA256}$ and the consensus mechanism of blockchain.

If both $S$ and $C$ faithfully follow the protocol, $S$ can pass the $\mathsf{Verify}$ algorithm and earn the service fee by presenting the preimage of the key to $C$ for decrypting the trained model, and $C$ can obtain the trained model or prediction service by posting the payment transaction.

\emph{Soundness.}
A forged training workload will be accepted by $C$ iff all the sampled iterations can pass the verification algorithm.
Because the outputs of VC can always be considered as the ground truth with the given inputs, a cheating server $S$ aims to provide a pair of models $\mathbf{w}_{i-1}$ and $\mathbf{w}_{i}$ such that the outputs of VC are the same with the commitment.
If the probability for achieving this is less than $2^{-l}$, the soundness property can be proved.
We show that if $S$ can pass $\mathsf{Verify}$ without the correct execution, there are only two ways for $S$ to cheat.

\textsc{Case 1. } The cheating $S$ directly forges the identifiers $\mathbf{I}_{i-1}$ and $\mathbf{I}_{i}$ without using the models when generating the commitment.
$S$ knows the randomly generated coefficient vector $\mathbf{v}$ before making the commitment.
We consider that $S$ can always find one of $\mathbf{w}_{i-1}^*$ and $\mathbf{w}_i^*$ which corresponds to the committed identifier, and he only needs to find the other one.
When the $i$-th iteration is sampled, since the forged identifiers are not computed from a model, the distribution of the identifiers is uniform (from the perspective of $S$).
If the bit lengths of both the output parameters $\mathbf{w}$ and the random coefficients $\mathbf{v}$ are $l$, their dot product extends the preimage of the hash to $2l$ bits, and the probability that the forged identifier is the same as the ground truth is $2^{-2l}$.

\textsc{Case 2. } We assume that the cheating $S$ can find a way to compute $\mathbf{w}_{i-1}^*$ or $\mathbf{w}_{i}^*$ without executing the training algorithm when generating the commitment.
We do not make any assumption about its approach, \eg, he may utilize some background knowledge about the model such as the distribution of the parameters.
Yet, according to the basic assumption that there exists no way to obtain the result without the correct computation,
the forged model and the ground truth should be different.
We also consider the scenario where $S$ can find one of $\mathbf{w}_{i-1}^*$ and $\mathbf{w}_{i}^*$ corresponding to the committed identifier.
Without loss of generality, the unknown model is denoted by $\mathbf{w}^*$.
To ensure that $\mathbf{w}^*$ and the ground truth $\mathbf{w}$ have the same identifier, we need to guarantee

    \begin{equation}\label{equ:proof4}
      \mathbf{v}_{1}(\mathbf{w}_{1}-\mathbf{w}_{1}^*)+\mathbf{v}_{2}(\mathbf{w}_{2}-\mathbf{w}_{2}^*)+\cdots+\mathbf{v}_{d}(\mathbf{w}_{d}-\mathbf{w}_{d}^*)=0.
    \end{equation}
Since $\mathbf{w}^*$ does not equal to $\mathbf{w}$, there exists $i$ such that
    \begin{equation}\label{equ:proof5}
    \mathbf{v}_{i} = \frac{\mathbf{x}(\mathbf{w}_{i}-\mathbf{w}_{i}^*)-\mathbf{v}_{i}(\mathbf{w}_{i}-\mathbf{w}_{i}^*)}{\mathbf{w}_{i}^*-\mathbf{w}_{i}}.
    \end{equation}

Because the coefficient vector $\mathbf{v}$ is uniformly random, the value of the right hand side in Equation~(\ref{equ:proof5}) is uniformly random. If the bit length of $\mathbf{w}$ is $l$, the probability that $\mathbf{v}_{i}$ satisfies Equation~(\ref{equ:proof5}) is $2^{-l}$.

There may exist some special cases which will discard some steps in an iteration. For example, in the SGD algorithm, if $S$ can find a model $\mathbf{w}^*$ whose outputs are exactly the same as the labels of the data batch, \ie, $\mathbf{w}^*\times\mathbf{x}=\mathbf{y}$, $\mathbf{w}^*$ will not be updated in this iteration. In other words, $S$ can directly output $\mathbf{w}^*$ without any computations. Such an attack is easy to be detected by VC, because $S$ knows the existence of the detection. Thus the soundness property of our system is satisfied.

If the claimed total number of iterations is $N$, the proportion of genuine identifiers is $t$ (\ie, the proportion is $1-t$ for corrupted ones), and $C$ randomly samples $c$ different iterations, the probability that all the sampled iterations are genuine is
  \begin{equation}\label{equ:sample}
    p = \frac{\tbinom{tN}{c}}{\tbinom{N}{c}} = \frac{(tN+1-c)(tN+2-c)\cdots(tN)}{(N+1-c)(N+2-c)\cdots N},
  \end{equation}
which means that the upper bound of this probability is $p$.

When a fake identifier passes the verification, this upper bound will be relaxed.
Given the definition of soundness, if all the data and parameters have $l$ bits, the probability that a fake identifier is exactly the same as the genuine one is less than $2^{-l}$.
Hence, the upper bound is relaxed to

\begin{equation}
 \epsilon = \sum_{i=0}^c\frac{\tbinom{tN}{c-i}}{\tbinom{N}{c}}\frac{1}{2^{li}}< p+(1-p)\sum_{i=1}^c\frac{1}{2^{li}}< p+\frac{1-p}{2^l-1}.
\end{equation}

For instance, if $S$ claims 100K iterations and performs only 70K iterations ($t=70\%$), $C$ only needs to verify $10$ or $14$ iterations to detect the misbehavior with a probability higher than $95\%$ or $99\%$, respectively.

\emph{Fairness.}
A malicious $S$ has incentives to forge proofs or identifiers in the commitments.
Since the program and the evaluation keys are provided by $C$, it is easy to detect fake proofs.
The probability that a fake commitment passes the verification is extremely small after sampling multiple iterations.
Therefore, $S$ is prevented from forging the training process, since once detected, it will not get paid.
Also, $S$ will faithfully present the preimage of the key $k$ to receive the payment.
As a result, the protocol ensures that $S$ will deliver the trained model and claim the real training workload.

A malicious $C$ is motivated to learn about the trained model without payment, and $C$ can only obtain the output of VC (\ie, the confidentiality of the witness depends on the underlying VC protocol).
It is also difficult for $C$ to manipulate $EK_F$ to infer the exact value of the witness.
Other malicious behaviors, such as claiming that a correct proof is incorrect or posting an invalid transaction on the blockchain, cannot help $C$ to learn the model.

\end{proof}
\begin{table*}[t!]
\centering
\begin{tabu}{c|c|c|c||c|c|c|c|c|c}
  \tabucline[1pt]{-}

Algorithm & Dataset & Dimension& Batch & $\mathsf{KeyGen}$ (s) & $\mathsf{Prove}$ (s) & $\mathsf{Verify}$ (ms) & $\mathsf{Native}$ (ms) & EK (MB) & VK (MB)  \\ \hline \hline
\multirow{3}*{Linear Regression} & \multirow{3}*{BHP} & \multirow{3}*{13} & 32 & 2.0 & 0.46 & 5.9 & 32 & 11.7 & 18.6 \\
\cline{4-10}
&& & 64 & 2.0 & 0.47 & 6.4 & 34 & 11.8 & 36.1 \\
\cline{4-10}
&& & 128 & 2.1 & 0.51 & 7.5 & 36 & 12.1 & 70.9 \\
\hline
\multirow{3}*{Logistic Regression}  & \multirow{3}*{US} & \multirow{3}*{15} & 32 & 2.1 & 0.42 & 6.1 & 33 & 11.9 & 21.1 \\
\cline{4-10}
&& & 64 & 2.2 & 0.45 & 6.2 & 34 & 12.1 & 41.0 \\
\cline{4-10}
&& & 128 & 2.2 & 0.47 & 7.6 & 37 & 12.5 & 80.1 \\
\hline
\multirow{3}*{NN} & \multirow{3}*{MNIST} & \multirow{3}*{784} & 32 & 25.4 & 12.0 & 989 & 475 & 159.1 & 19.5 \\
\cline{4-10}
&& & 64 & 27.0 & 13.4 & 1068 & 825 & 166.7 & 21.2 \\
\cline{4-10}
&& & 128 & 33.4 & 16.7 & 1231 & 1523 & 185.9 & 24.4 \\
\hline
\multirow{3}*{SVM} & \multirow{3}*{US} & \multirow{3}*{15} & 32 & 2.1 & 0.43 & 6.1 & 94 & 11.6 & 24.2 \\
\cline{4-10}
&& & 64 & 6.4 & 0.45 & 6.5 & 103 & 11.7 & 45.4\\
\cline{4-10}
&& & 128 & 2.2 & 0.47 & 7.5 & 115 & 11.8 & 87.8 \\
\hline
\multirow{3}*{K-Means} & \multirow{3}*{MNIST} & \multirow{3}*{784}& 32 & 8.7 & 2.9 &42.3 & 334 & 52.4 & 1.6 \\
\cline{4-10}
&& & 64 & 15.3 & 5.4 & 64.5 & 601 & 99.4 & 2.5 \\
\cline{4-10}
&& & 128 & 27.3 & 9.6 & 106 & 1253 & 189.5 & 4.5 \\
\hline
 Decision Tree & Nursery & 27 & $-$ & 6.2 & 1.8 & 5.5 & 136 & 11.9 & 0.01   \\
  \tabucline[1pt]{-}
\end{tabu}
\caption{Verification costs of all implemented algorithms}
\label{table:all}
\end{table*}

\textbf{Comparisons.} Several other attacks have been proposed to break the original ZKCP protocol~\cite{campanelli2017zero}. First, if the purchase is a service (not goods), \eg, an audit of online file storage, $C$ can infer from the proof that the service is correct, and then abort the protocol without payment. However, for the outsourced ML service, the client aims to obtain the trained model or prediction results rather than a simple answer of yes or no. The proof only allows $C$ to certify the correctness of the service but will not reveal any additional information. Hence, there is no incentive for $C$ to abort the protocol after verifying the proofs.

Second, $C$ may modify the common reference string (CRS), \ie, the $EK_F$ and $VK_F$, to learn information from proofs~\cite{campanelli2017zero}. A malicious $C$ can check whether a value in the witness is the exact one or not. For instance, in the pay-to-sudoku service, the client can find out the exact value for a Sudoku cell with a probability of $1/9$. However, for machine learning services, it is difficult to find the values in witness as they are represented by more bits. Specifically, if a parameter in the witness has $l$ bits, the probability of finding its exact value is $1/2^l$, which is negligible. Therefore, $C$ is allowed to choose the CRS in our scheme.

Finally, $C$ needs to ensure that the results delivered by $S$ are actually encrypted by the key. This may require symmetric encryption schemes such as AES~\cite{lindell2014introduction}, and the circuit implementation of which will lead to a high cost ($4.2\times 10^6$ constraints only for 300 blocks~\cite{kosba2018xjsnark}). For instance, the simple three-layer fully-connected neural network with 128 hidden neurons consumes 25,408 blocks (if each parameter has 32 bits), which is too expensive. Recently, FairSwap~\cite{dziembowski2018fairswap} proposed to utilize the Merkle tree to verify the hash value of a large file without using the zero-knowledge proof. Combining this technique with VeriML may help achieve stronger security guarantee.

In summary, the attacks which can break the ZKCP protocol will not affect our scheme, because it is difficult to guess the continuous values in machine learning, and what we focus on is the verification of the computation workload of the training process, rather than verifying the existence of just a service or a file. Moreover, because the underlying VC protocol is independent with our system, these attacks can be better defended by other building blocks which do not need the trusted setup, such as~\cite{chiesa2017zero,wahby2015efficient}.

\textbf{Discussions about the DoS attack.} Apart from the fairness concern to ensure a fair exchange of service and payment between the client and the server, another issue is that parties may prematurely abort from the protocol. For instance, a malicious server may be interested in obtaining the data from the client but abort before or during the computation phase, while a malicious client may launch a DoS attack by aborting the protocol before the payment phase or lying about the verification to consume the computing resources of the server.

One possible solution is to ask the client and the server to make certain deposits in advance. If the protocol is followed through, each party will be refunded their deposits; if one party prematurely aborts the protocol, this party will lose the deposit. Intuitively, the deposit of the malicious party should be used to compensate for the other party. However, it is difficult for the smart contract to detect certain malicious behaviors, \eg, if the client outputs $\textsf{Reject}$ in the verification phase, it is hard to tell whether the server fails in the verification or the client lies about the verification. This is partly because the cost of implementing complex operations on a smart contract is very expensive. In~\cite{dziembowski2018fairswap}, the authors proposed to mitigate the risk from a malicious client by making the server do some pre-computation activities. But in MLaaS, the server cannot perform any pre-computation before having received the data from the client. Therefore, the traditional deposit and refund strategy cannot well solve the malicious abort problem in MLaaS, and this can be an interesting open problem.

\textbf{The correctness of the delivery. }
In this paper, we are mainly concerned about the  correctness of the delivery, \ie., ${S}$ will not deliver a fake model if the integrity of the training has been verified, since this behavior brings no economic incentives, and may affect the server's reputation potentially. To provide stronger security guarantee, here we discuss how to verify the correctness of the delivered model.

The goal of verifying the model correctness is to ensure that the delivery has the claimed prediction accuracy, $S$ will first construct a circuit, taking the delivery as input, to make predictions on a small test set. Then, $S$ extends this circuit to commit the hash of the model rather than its identifier, and produces the corresponding proof. Using Merkle hash tree and the FairSwap protocol can help save the computation overhead of this step. If this proof can pass the verification algorithm, the correctness of the delivered model is verified. For the prediction service, this step is more effective since the prediction results may have smaller size and can be hashed by the circuit directly.

Another potential malicious behavior is that after attaining the pre-defined termination condition $S$ may keep training to charge more. The use of VC to verify the performances of multiple epochs might too expensive. We concern that over-training also requires $S$ to execute the computation which causes low benefits, so appropriate billing policies which can motivate $S$ to finish a task as soon as possible can mitigate such this behavior. For example, if the training task is finished at a given time, $S$ will earn an additional bonus and a better reputation.

\begin{figure*}[!th]
\begin{minipage}[t]{0.49\columnwidth}
\centering
\includegraphics[width=1\columnwidth]{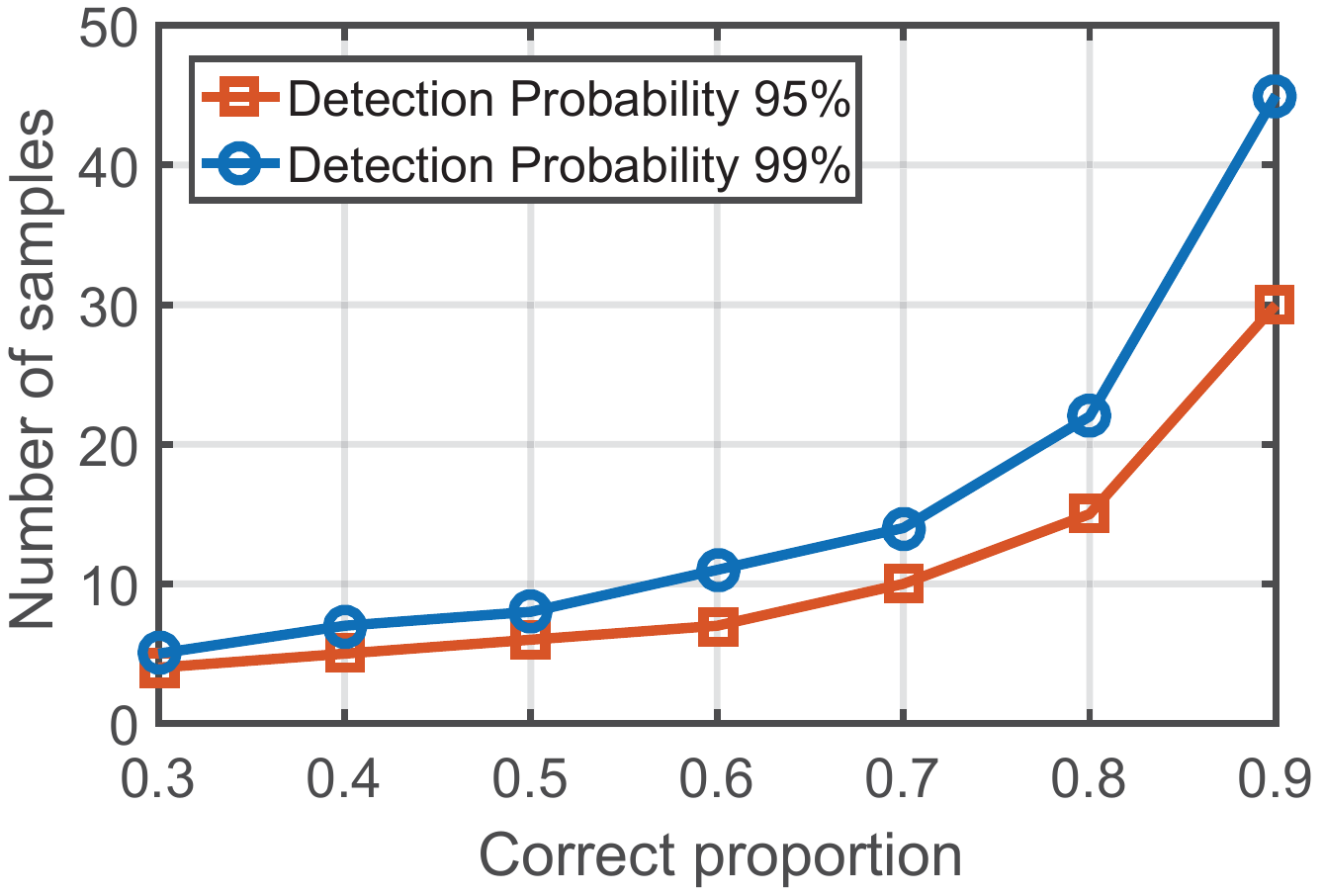}

\caption{Sampling strategy}
\label{fig:sample}
\end{minipage}
\begin{minipage}[t]{0.49\columnwidth}
\centering
\includegraphics[width=1\columnwidth]{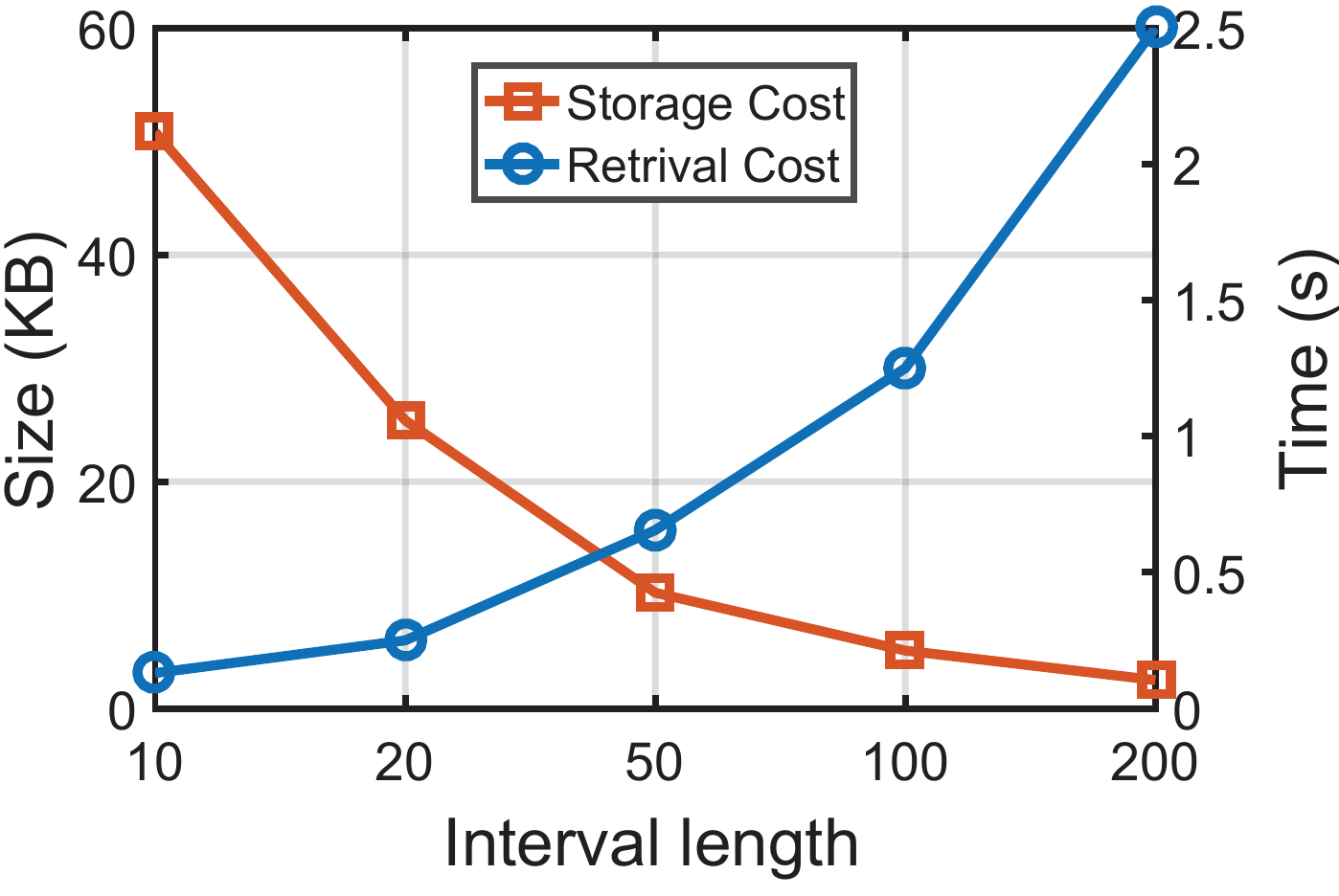}

\caption{Impact of interval length}
\label{fig:interval}
\end{minipage}
\begin{minipage}[t]{0.49\columnwidth}
\centering
\includegraphics[width=1\columnwidth]{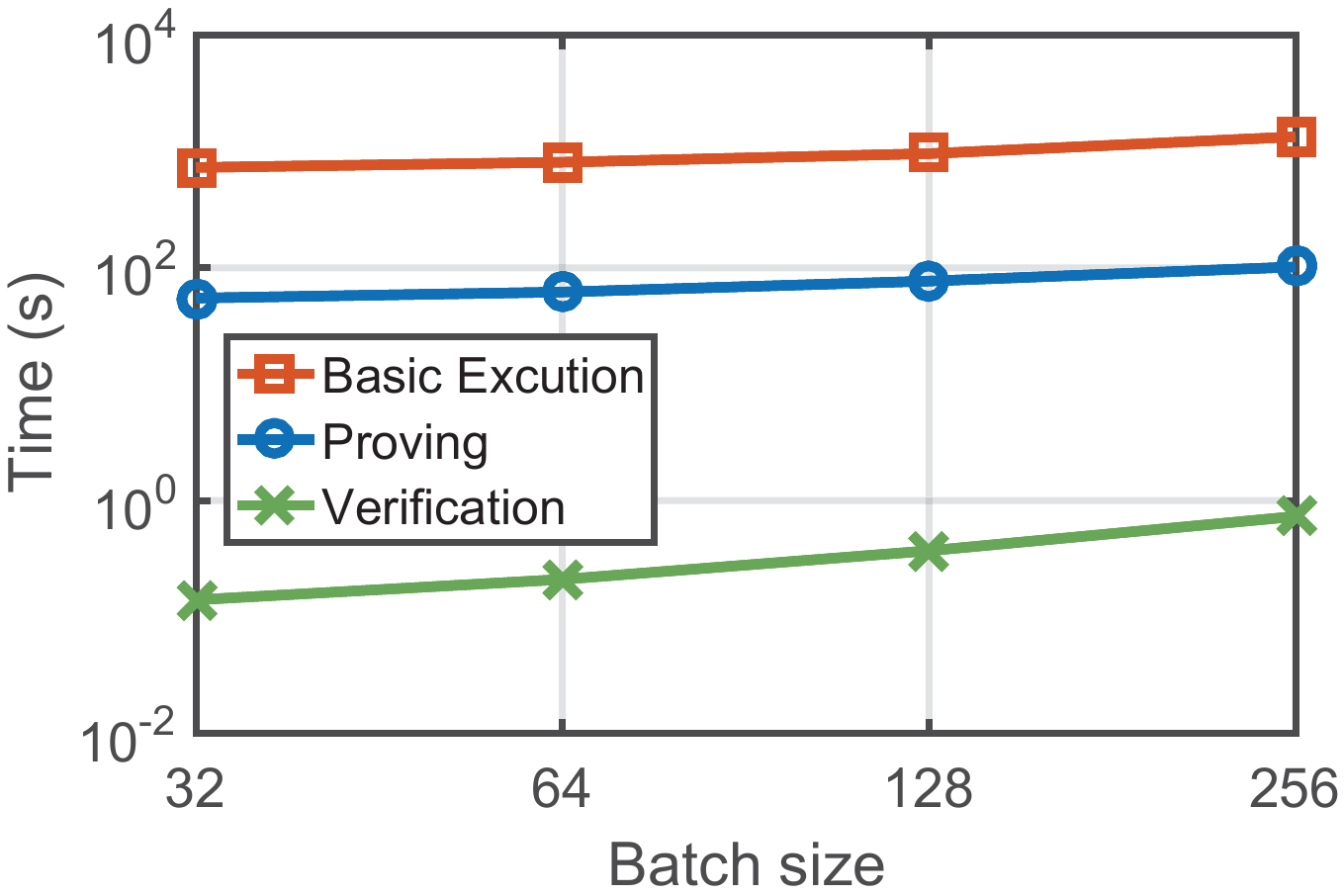}

\caption{Run-time of verifying linear regression}
\label{fig:time}
\end{minipage}
\begin{minipage}[t]{0.49\columnwidth}
\centering
\includegraphics[width=1\columnwidth]{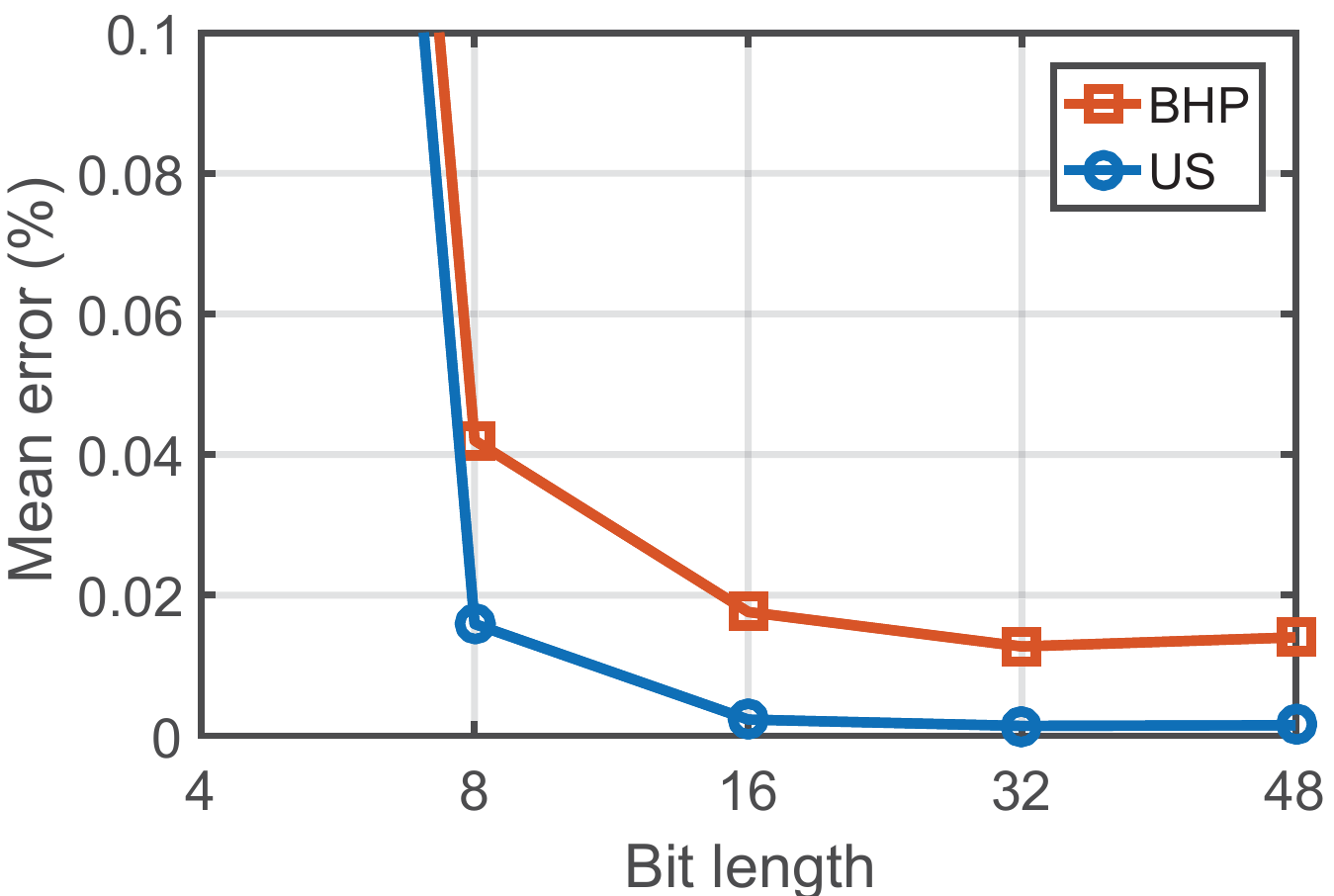}

\caption{Impact of bit length}
\label{fig:bit_loss}
\end{minipage}

\end{figure*}

\section{Performance Evaluation}\label{sec:eva}
In this section, we present the implementation and experimental results to show the performance of VeriML.

\subsection{Implementations}
\textbf{Setup.} Our system is implemented in Java. We use the jsnark compiler~\cite{jsnark} to produce the circuit and adopt libsnark~\cite{libsnark} as the backend. The server side which executes the proving part runs on a desktop with Ubuntu 16.04, Intel Xeon W-2133 CPU and 64GB RAM. The client side runs on a laptop with Ubuntu 16.04, Intel Core i5-4460S CPU and 16GB RAM to execute the generation and verification.

\textbf{Datasets.} The datasets we use are as follows. The Boston House Price dataset (BHP)~\cite{housing} contains 506 samples, 13 features, and the label is the house price. The banknote authentication (BA) dataset is extracted from images of genuine and forged banknote-like specimens with 4 features and a binary label~\cite{banknote}. The Nursery dataset~\cite{nursery} has 12,960 samples and 8 features (we bucket the features to 27 bins), and the label is the outcome of Slovenian nursery admission process. The dataset US~\cite{Minne} contains 600K census records with 15 features from the United States. MNIST has 60K images of handwritten digits, each with 784 features~\cite{mnist}.

In VeriML, we implemented six popular machine learning algorithms: linear regression, logistic regression, neural network, support vector machine, K-Means and decision tree. Concretely, we implemented decision tree based on the CART algorithm~\cite{breiman2017classification}, and we implemented a 4-layer fully-connected neural network which is the same as~\cite{mohassel2017secureml, mohassel2018aby} with two hidden layers (each hidden layer has 128 neurons) and the approximated sigmoid activation function, and the loss function is the cross entropy function. All the other algorithms are implemented based on the standard versions.

From the point of efficiency view, we calculated the running time and communication overhead by tuning the batch size, the main parameter which can affect the performance. As a comparison, we evaluate the time of the native execution by running the constructed circuit directly in jsnark, \ie, the time of only executing the essential computations. We do not choose to compare with existing ML packages, because the computation processes without fixed-point number representations are incompatible with VC techniques, and the plenty of optimization methods in these packages have not been implemented by circuits yet. Furthermore, we evaluate the effect of our sampling and checkpoint strategy. From the accuracy point of view, we evaluated the effect of fixed-point approximation, and compared all the proposed approximated activation functions with benchmarks and existing methods on 4 real-world datasets.

\subsection{Computation Overhead}
Because our design is based on sampling which randomly chooses multiple iterations for verification, we first present the impact of the sampling strategy, and then discuss the usability of our system in practice. We assume that both the server and the client store the dataset and compile the circuits in advance.

Since the proportion of forged iterations is unknown to the client, the number of sampled iterations is determined by the client's expected confidence. Figure~\ref{fig:sample} presents the needed samples for different proportions of the correct-executed iterations. The x-axis is the expected correct proportion of all iterations, and the y-axis is the number of challenges. When the proportion of forged iterations is large, verifying a few iterations will enable the client to detect whether the commitment is forged. With the increasing proportion of genuine iterations, the required number of challenges will rapidly increase and cause high latency.

In practice, the server has incentives to forge a lot more iterations to gain large enough profits, which helps to introduce an important observation: the number of challenges is affected by the proportion of genuine iterations, but, insensitive to the total number of claimed iterations (because in Equation~(\ref{equ:sample}), $N$ is much larger than $c$), thus around 10 to 15 challenges are appropriate.

Table~\ref{table:all} shows the costs of each verification in all the implemented algorithms. Obviously, to make the ML service available in practice, the system needs to ensure that the verification overhead on the client side is less than executing the training locally. If the claimed number of iterations is $N$, the interval length is $m$, the number of challenges is $c$, the running time of $\mathsf{KeyGen}$ is $t_k$, and the proving time, verifying and executing one iteration are $t_p$, $t_v$ and $t_e$ respectively. Because the overhead of executing $\mathsf{KeyGen}$ is one-off for each task, we have $t_k + ct_v < Nt_e$. According to Table~\ref{table:all}, VeriML is practicable when the task has hundreds of iterations.

The costs of SNARK grows linearly with the circuit size, \ie, the number of inputs and multiplications of the circuit. Therefore, batch size is a major factor affecting the overhead. Table~\ref{table:all} shows that the overhead grows linearly with the batch size. Note that training decision tree does not require to set the batch size, so we omit the discussion here.

On the server side, to make the system as economical as possible, we also expect to prove that the overhead is less than executing the training locally. So we have $\frac{m}{2}ct_e + ct_p < Nt_e$, \ie, the total time of generating proofs and retrieving the model's states is less than that of training. For the implemented algorithms, VeriML is economic when the task has thousands of iterations.

The latency is approximately equivalent to the total time of retrieving states, producing and verifying proofs, which can be written as $\frac{m}{2}ct_e + c(t_p + t_v)$. Using the linear regression task as an example, we can see from Table~\ref{table:all} that the time of retrieving one state is estimated to be 0.9s while setting the interval length is 50, and the whole checkpointing scheme requires about 11.7s with 13 challenges. The time costs of proving and verification are 6.5s and 0.1s respectively. Furthermore, the server has additional storage costs in retrieving the model's state from checkpoints. If the bit length of parameters is $l$, and the model's dimension is $d$, the storage cost can be calculated as $ldN/m$ bits. If the total number of iterations is 10K, the storage cost is 10.2KB. Figure~\ref{fig:interval} shows the effects of interval length on the storage overhead and retrieving time. We can see that a bit length of 50 is an appropriate balance of computation and storage costs.
\begin{table*}[t!]
\centering
\begin{tabu}{c|c|c||c|c|c|c|c|c}
  \tabucline[1pt]{-}
Algorithm & Dataset & Dimension &  $\mathsf{KeyGen}$ (s) & $\mathsf{Prove}$ (s) & $\mathsf{Verify}$ (ms) & $\mathsf{Native}$ (ms) & EK (MB) & VK (MB)  \\ \hline \hline

Linear Regression & BHP & 13 & 0.17 & 0.08 & 6.2 & 15 & 0.2 & 0.02 \\
\hline
Logistic Regression  & US & 15 & 0.2 & 0.1 & 6.2 & 21 & 0.2 & 0.02 \\
\hline
NN  & MNIST & 784& 22.9 & 11.7 & 186.7 & 408 & 39.8 & 2.3 \\
\hline
SVM & US & 15 & 0.2 & 0.1 & 7.7 & 21 & 0.2 & 0.02  \\
\hline
K-Means & MNIST & 784 & 38.6 & 23.8 & 123.3 & 356 & 93.9 & 1.6 \\
\hline
Decision Tree & Nursery & 8 & 0.2 & 0.04 & 5.0 & 6 & 0.26 & 0.01   \\
  \tabucline[1pt]{-}
\end{tabu}
\caption{Verification costs of computing predictions for implemented algorithms}
\label{table:pred_all}
\end{table*}

\begin{table}[t]
    \vspace{2mm}
  \centering
  \begin{tabu}{c|c ccc}
    \tabucline[1pt]{-}
       {Dataset} & Sigmoid & Taylor & Remez & Piecewise  \\
      \hline
      \hline
    BN & $\mathbf{73.41}$ & 72.60 & $\mathbf{73.24}$ & 73.29  \\
    US & $\mathbf{87.81}$ & 85.11 & $\mathbf{86.17}$ & 85.84 \\

    MNIST & $\mathbf{95.49}$ & 87.82 & $\mathbf{95.58}$ & 96.15  \\
     \tabucline[1pt]{-}
  \end{tabu}
  \caption{Accuracy of approximated sigmoid}
  \label{tab:sigmoid}
  \end{table}
Figure~\ref{fig:time} shows the runtime of verifying the whole training process using linear regression as a concrete example. Here we assume that the claimed number of iterations is 100K, the dimension is 13, the batch size is 128, the interval length is 50, and the proportion of genuine iterations is 70\%. The server can prove the claimed workload by spending about 2.2\% of the native execution time, and the cost of verification by the client is far less than the re-execution of the training task.

Naturally, if the learning task only takes a few iterations, the cost and latency of verification may not be economic for the server and the client. But if the convergence speed is slow, \ie, the number of iterations is large, the verification is worthwhile.

Table~\ref{table:pred_all} shows the performance of verifying making predictions over batch size as 64 for all the algorithms. Compared with Table~\ref{table:all}, computing predictions is much cheaper than one iteration in training. This is due to: (1) the part of hashing is removed from the circuit, (2) it does not need to add the additional random coefficient vector, and (3) some steps are no longer needed in making predictions, such as the backward propagation part.

\subsection{Communication Overhead}
The communication overhead lies in transmitting the proofs and the evaluation keys.

For the first part, because the sizes of each proof and each identifier are constant (288 bytes and 256 bits respectively), the communication overhead is only related to the number of iterations and samples, and will be affected by the kind of the algorithm. For example, with 10K iterations and 15 proofs, it only costs about 317KB. The communication cost of the prediction service is constant because only one proof is required, \ie, 288 bytes without the delivery of prediction results.

For the second part, the size of $EK_F$ is linear with the number of inputs and multiplications of the circuit. Table~\ref{table:all} shows the sizes of the keys for different algorithms. The size of $EK_F$ may be much larger than those of commitments and proofs. Fortunately, it can be reused for all samples in each task and transmitted in parallel during the training. Therefore, it will not affect the performance of verification.

Furthermore, the cost of executing a smart contract to exchange the key is also cheap. In Ethreum, evaluating one hash function spends 27,265 gas, which can be translated to only 0.000214 Ether or 0.047 USD for an exchange rate of 220 USD/Ether.

\subsection{Accuracy Loss}
The accuracy loss in our system can be analyzed from two aspects: the fixed-point representation of rational numbers, and approximate functions.

In Figure~\ref{fig:bit_loss}, we use linear regression as an example to show the impact of the fixed-point representation. The learning rate is set as 0.05 and the number of iterations is 10K. While varying the bit length between 4 and 48, the mean error of the trained model decreases rapidly. We can observe that a 32-bit representation has no obvious accuracy loss. In our implementations, linear regression, SVM and K-Means do not include the approximate functions, so their accuracy losses can be ignored. Furthermore, the prior work~\cite{ke2017lightgbm} shows that using histogram will not affect the accuracies of decision trees.

For implementing logistic regression and neural network, we introduce the approximated sigmoid function as the activation function. Table~\ref{tab:sigmoid} shows that the Remez method has better performance than the Taylor extension and piecewise methods.

Recitified Linear Unit (ReLU) $f(x)=\max(0, x)$ is another popular activation function in neural networks. Applying ReLU in arithmetic circuits also incurs a huge computation cost similar to sigmoid functions. If one layer uses ReLU as the activation function, it needs to execute $bn_i$ comparisons, which are even more expensive than verifying the matrix multiplication. Thus, the square function $f(x)=x^2$ might be suitable for replacing ReLU in VeriML~\cite{gilad2016cryptonets}.
Prior results show that using the square function to replace ReLU can achieve satisfactory performance (99\% accuracy)~\cite{gilad2016cryptonets, ghodsi2017safetynets}. However, using it to approximate softmax fails with the fixed learning rate. But we notice that it still performs well (93.09\% accuracy) when we use the Adam Optimizer to adjust the learning rate.

\section{Conclusion}
In this paper, we presented the design, implementation and evaluation of VeriML, a verifiable and fair outsourced machine learning service. We transformed machine learning algorithms to quadratic arithmetic circuits to generate proofs, then designed a new commit-and-prove protocol to detect misbehaviors during the training process with high probability. Blockchain is leveraged to serve as a decentralized trusted third party to achieve fair exchange by a hash-locked transaction. Our experiments on real-world datasets validate that the computation and communication costs of VeriML are practical, and it can be readily applied to the existing MLaaS platforms.

There are some important future directions. First of all, the efficiency of existing SNARK implementations can be further improved. The underlying VC protocol can be replaced by other similar building blocks which have the zero-knowledge property. Moreover, we will seek for better approximations of non-linear functions, especially for the diversified non-linear functions used in many other machine learning algorithms.

\bibliographystyle{IEEEtran}
\bibliography{reference}

\end{document}